\documentclass[a4paper,aps,prl,showpacs,twocolumn,superscriptaddress]{revtex4-1}   %APS document class style  (PRA,PRL,ETC)
   
\usepackage[utf8]{inputenc}  %Input what you want e.g., é, ł, a, ü
\usepackage[T1]{fontenc}     %Output what you want e.g., é, ł, a, ü
\usepackage[british]{babel}  %Do hyphenation according to british english
\usepackage{lmodern}  %Joe's font

\usepackage[scaled=1.03]{inconsolata} %Monospace font
\usepackage[usenames,dvipsnames]{color} %More color choices
\usepackage[colorlinks,citecolor=blue,linkcolor=magenta,urlcolor=blue]{hyperref}  %Hyperlinks (pink, green, blue)
\usepackage{graphicx} % Package to insert external figures
\usepackage{tikz}
\usepackage[babel]{microtype}  %Improves text justification
\usepackage{amsmath,amssymb,amsthm,bm,mathtools,amsfonts,mathrsfs,bbm,dsfont} %Useful math packages
\usepackage{xspace}  %Useful to add space in macros
\usepackage{multirow}
\usepackage{verbatim}
\usepackage{enumerate}
\usepackage{tcolorbox}
\usepackage[normalem]{ulem}

\usepackage{physics}
\newcommand{\id}{\ensuremath{\mathds{1}}}
\usepackage{bbold}
\usepackage[mathscr]{eucal}

\newcommand{\Sys}{\textrm{S}} % system
\newcommand{\Anc}{\textrm{A}} % ancilla
\newcommand{\Env}{\textrm{E}}
\newcommand{\Dyn}{\mathcal{D}}

\newcommand{\kb}[2]{|#1\rangle\langle#2|} %ketbra
\newcommand{\M}{\mathsf{M}}

\def\e{\ensuremath{\mathrm{e}}}
\def\i{\ensuremath{\mathrm{i}}}
\newcommand{\cpt}{\mathcal{E}}
\newcommand{\dynamics}[1]{(#1)}

\newtheorem{definition}{Definition}
\newtheorem{theorem}{Theorem}

\newcommand{\eoa}{E^\sharp}
\newcommand{\eof}{E}
\newcommand{\coa}{\mathcal{C}^\sharp}
\newcommand{\cof}{\mathcal{C}}
\newcommand{\choi}{\chi}

\usepackage{bbold}
\newcommand{\one}{\mathbb{1}}
\newcommand{\damp}{\kappa}

\newcommand{\maptad}[1]{{\mathcal{A}_{#1}}}

\newcommand{\coupling}{\alpha}

\renewcommand{\d}{\mathrm{d}}
\newcommand{\placeholder}{\rho}
\newcommand{\gammashort}{\nu}

\usepackage{tikz}
\usetikzlibrary{patterns}
\usetikzlibrary{patterns.meta}

\tikzdeclarepattern{
  name=hatch,
  parameters={\hatchsize,\hatchangle,\hatchlinewidth},
  bounding box={(-.1pt,-.1pt) and (\hatchsize+.1pt,\hatchsize+.1pt)},
  tile size={(\hatchsize,\hatchsize)},
  tile transformation={rotate=\hatchangle},
  defaults={
    hatch size/.store in=\hatchsize,hatch size=5pt,
    hatch angle/.store in=\hatchangle,hatch angle=0,
    hatch linewidth/.store in=\hatchlinewidth,hatch linewidth=.4pt,
  },
  code={
      \draw[line width=\hatchlinewidth] (0,0) -- (\hatchsize,\hatchsize);
  }
}

\tikzdeclarepattern{
  name=mylines,
  parameters={
      \pgfkeysvalueof{/pgf/pattern keys/size},
      \pgfkeysvalueof{/pgf/pattern keys/angle},
      \pgfkeysvalueof{/pgf/pattern keys/line width},
  },
  bounding box={
    (0,-0.5*\pgfkeysvalueof{/pgf/pattern keys/line width}) and
    (\pgfkeysvalueof{/pgf/pattern keys/size},
0.5*\pgfkeysvalueof{/pgf/pattern keys/line width})},
  tile size={(\pgfkeysvalueof{/pgf/pattern keys/size},
\pgfkeysvalueof{/pgf/pattern keys/size})},
  tile transformation={rotate=\pgfkeysvalueof{/pgf/pattern keys/angle}},
  defaults={
    size/.initial=5pt,
    angle/.initial=45,
    line width/.initial=.4pt,
  },
  code={
      \draw [line width=\pgfkeysvalueof{/pgf/pattern keys/line width}]
        (0,0) -- (\pgfkeysvalueof{/pgf/pattern keys/size},0);
  },
}

\usetikzlibrary{arrows.meta, shadings}

\usepackage{pgfplots}
\usepackage{tikz-3dplot}
\tdplotsetmaincoords{60}{115}
\tdplotsetrotatedcoords{25}{-25}{0}
\pgfplotsset{compat=newest}
\usetikzlibrary{arrows.meta, arrows, positioning}

\pgfplotsset{compat = newest}

\newcommand{\blochdyn}[8]{
	\draw (1,0, 0) node[anchor=north west] {$x$};
	\draw (0, 0.95, 0.2) node[anchor=north west] {$y$};
	\draw (0, 0, 1) node[anchor=north west] {$z$};
	
	\shade[ball color = lightgray,opacity = 0.2] (0,0,0) circle (1cm);
	\tdplotsetrotatedcoords{#6}{#7}{#8};
	
	\shade[tdplot_rotated_coords,right color=orange,middle color=red,left color=blue,opacity=0.3,shading angle=-110] (#1,#2,#3) circle (#4 and  #5);
	\shade[tdplot_rotated_coords,ball color=orange,opacity=0.2] (#1,#2,#3) circle (#4 and #5);
	
	\tdplotsetrotatedcoords{0}{0}{0};
	\draw[dashed,tdplot_rotated_coords,gray] (0,0,0) circle (1);
	
	\tdplotsetrotatedcoords{90}{90}{90};
	\draw[dashed,tdplot_rotated_coords,gray] (1,0,0) arc (0:360:1);
	
	\tdplotsetrotatedcoords{0}{90}{90};
	\draw[dashed,tdplot_rotated_coords,gray] (1,0,0) arc (0:360:1);
	
	\draw[-{Latex[length=2mm]}] (0,0,0) -- (1,0,0);
	\draw[-{Latex[length=2mm]}] (0,0,0) -- (0,1,0);
	\draw[-{Latex[length=2mm]}] (0,0,0) -- (0,0,1);
	\draw[dashed, gray] (0,0,0) -- (-1,0,0);
	\draw[dashed, gray] (0,0,0) -- (0,-1,0);}

\newcommand{\blochdynwith}[8]{
	\draw (1,0, 0) node[anchor=north west] {$x$};
	\draw (0, 0.95, 0.2) node[anchor=north west] {$y$};
	\draw (0, 0, 1) node[anchor=north west] {$z$};
	\draw (0.2, -0.185, -1) node[anchor=north west] {$\ket{0}$};
	\draw (0,-0.34,2) node[anchor=north west] {$\ket{1}$};
	
	\shade[ball color = lightgray,opacity = 0.2] (0,0,0) circle (1cm);
	\tdplotsetrotatedcoords{#6}{#7}{#8};
	
	\shade[tdplot_rotated_coords,right color=orange,middle color=red,left color=blue,opacity=0.3,shading angle=-110] (#1,#2,#3) circle (#4 and  #5);
	\shade[tdplot_rotated_coords,ball color=orange,opacity=0.2] (#1,#2,#3) circle (#4 and #5);
	
	\tdplotsetrotatedcoords{0}{0}{0};
	\draw[dashed,tdplot_rotated_coords,gray] (0,0,0) circle (1);
	
	\tdplotsetrotatedcoords{90}{90}{90};
	\draw[dashed,tdplot_rotated_coords,gray] (1,0,0) arc (0:360:1);
	
	\tdplotsetrotatedcoords{0}{90}{90};
	\draw[dashed,tdplot_rotated_coords,gray] (1,0,0) arc (0:360:1);
	
	\draw[-{Latex[length=2mm]}] (0,0,0) -- (1,0,0);
	\draw[-{Latex[length=2mm]}] (0,0,0) -- (0,1,0);
	\draw[-{Latex[length=2mm]}] (0,0,0) -- (0,0,1);
	\draw[dashed, gray] (0,0,0) -- (-1,0,0);
	\draw[dashed, gray] (0,0,0) -- (0,-1,0);}

\definecolor{mathematicablue}{rgb}{0.87,0.94,1}
\definecolor{mathematicadarkblue}{rgb}{0.368417, 0.506779, 0.709798}
\definecolor{mathematicaorange}{rgb}{1,0.9,0.8}
\definecolor{mathematicadarkorange}{rgb}{1,0.5,0}

\newcommand{\twoqubgate}[4]{
	\draw[line width=1.2pt, color=mathematicadarkblue, opacity=1, rounded corners=2pt, fill=mathematicablue, fill opacity=1] (#1, -1.65-#2) rectangle ++(#4,1.5) node[midway, color=black, opacity=1]{#3};
}

\newcommand{\qwire}[3]{
	\draw[line width=1.2pt,, color=mathematicadarkblue, opacity=1] (#1, #2) -- (#1+#3, #2);
}

\newcommand{\qwireup}[4]{
	\draw[line width=1.05pt] (#1, #2) -- (#3, #4);
}

\newcommand{\memoryarrow}[3]{
\draw [-latex, line width=1.05pt] (#1,#2) -- (#1,-0.5+#2) -- (#1+#3, -0.5+#2) --(#1+#3, -0+#2);
}

\newcommand{\measurement}[2]{
\node [line width=1.05pt] (36) at (0.125+#1, -0.15+#2) {};
\node [line width=1.05pt] (37) at (0.875+#1, -0.15+#2) {};
\node [line width=1.05pt] (38) at (0.5+#1, -0.25+#2) {};
\node [line width=1.05pt] (39) at (0.85+#1, 0.25+#2) {};
\draw [line width=1.05pt, bend left=60, looseness=1.10] (36.center) to (37.center);
\draw [-latex, line width=1.05pt] (38.center) to (39.center);
\draw[line width=1.05pt] (#1,-0.35-#2) rectangle ++(1,0.7);
}

\makeatletter
\def\maketitle{
\@author@finish
\title@column\titleblock@produce
\suppressfloats[t]}
\makeatother

\begin{document}

\author{Charlotte Bäcker}
\affiliation{Institute of Theoretical Physics, TUD Dresden University of Technology, 01062, Dresden, Germany}
\author{Konstantin Beyer}
\affiliation{Institute of Theoretical Physics, TUD Dresden University of Technology, 01062, Dresden, Germany}
\affiliation{Department of Physics, Stevens Institute of Technology, Hoboken, New Jersey 07030, USA}
\author{Walter T. Strunz}
\affiliation{Institute of Theoretical Physics, TUD Dresden University of Technology, 01062, Dresden, Germany}

\title{Local disclosure of quantum memory in non-Markovian dynamics}
\date{\today}

\begin{abstract}
%Abstract \fixme{Neue internationale Abkürzung für die TUD notwendig bei den affiliations?
%}
Non-Markovian processes may arise in physics due to memory effects of environmental degrees of freedom. For quantum non-Markovianity, it is an ongoing debate to clarify whether such memory effects have a verifiable quantum origin, or whether they might equally be modeled by a classical memory. In this contribution, we propose a criterion to test locally for a truly quantum memory. 
The approach is agnostic with respect to the environment, as it solely depends on the local dynamics of the system of interest. 
Experimental realizations are particularly easy, as only single-time measurements on the system itself have to be performed.
We study memory in a variety of physically motivated examples, both for a time-discrete case, and for time-continuous dynamics. For the latter, we are able to provide an interesting class of non-Markovian master equations with classical memory that allows for a physically measurable quantum trajectory representation.
\end{abstract}

\maketitle

\paragraph{Introduction---}
Applying quantum technologies to real-world problems requires a fundamental understanding of all underlying physical processes. Possible quantum advantages rely on our ability to cope with noise and dissipation, induced by the environment~\cite{knillResilientQuantumComputation1998,harrowQuantumComputationalSupremacy2017,preskillQuantumComputingNISQ2018,roffeQuantumErrorCorrection2019a,bruzewiczTrappedionQuantumComputing2019,grumblingQuantumComputingProgress2019,wangNoiseinducedBarrenPlateaus2021}.
%, emerging from interactions with the environment~\cite{}.
%In the simplest models these environmental influences induce a memoryless (Markovian) dynamics in the system of interest. 
%Clearly, for further progress in controlling quantum systems, it is increasingly important to characterize the physics of the environment in more detail.
%This entails he impact of memory effects induced by the environment, rendering the dynamics non-Markovian.
A detailed modeling of environmental impacts entails memory effects, showing \emph{non-Markovianity}~\cite{vacchiniMarkovianityNonMarkovianityQuantum2011, breuerFoundationsMeasuresQuantum2012, rivasQuantumNonMarkovianityCharacterization2014, hallCanonicalFormMaster2014, breuerColloquiumNonMarkovianDynamics2016,liConceptsQuantumNonMarkovianity2018, breuerMeasureDegreeNonMarkovian2009,rivasEntanglementNonMarkovianityQuantum2010}.
This requires advanced methods to describe quantum devices, yet non-Markovianity might also help to mitigate errors \cite{oreshkovContinuousQuantumError2007, manHarnessingNonMarkovianQuantum2015, weiCharacterizingNonMarkovianOffResonant2023}.

In recent years, it has become evident that non-Markovianity in quantum dynamics need not have a quantum origin \cite{vacchiniClassicalAppraisalQuantum2012, megierEternalNonMarkovianityRandom2017,filippovDivisibilityQuantumDynamical2017,megierMemoryEffectsQuantum2021}.
The ability to distinguish memory effects arising from the coupling to an environmental quantum system from those of classical nature is of fundamental importance. On the one hand, it will help to improve the performance of quantum devices, as error-correction schemes differ in the two cases. On the other hand, such studies are inevitable when trying to prove the quantum nature of unfathomable degrees of freedom such as gravity~\cite{carneyUsingAtomInterferometer2021,maLimitsInferenceGravitational2022}.

%Assume we have a quantum system of interest, denoted by \(\Sys\), which is embedded in an environment \(\Env\). 
Clearly, full operational access to the environment reveals its quantum nature, a situation hardly met in experiments. Indeed, standard open system theory aims at an effective dynamical description of the system of interest \(\Sys\), without any explicit reference to the environment \(\Env\).
Accordingly, we assume throughout that
%In such an environment-agnostic approach one relies only on 
information is available from measurements on \(\Sys\) only. The question arises whether such local information suffices to distinguish memory effects induced by an unknown quantum environment from those that may arise classically.
%We then may ask again: If we have only information that is locally available on \(\Sys\), is it possible to distinguish memory effects induced by a quantum environment from those that may arise classically?

Recently, this question has been addressed in the framework of process tensors~\cite{milzWhenNonMarkovianQuantum2020,giarmatziWitnessingQuantumMemory2021,nerySimpleMaximallyRobust2021,tarantoCharacterisingHierarchyMultitime2023}. A process tensor bears all information about the statistics of any possible sequence of measurements that could be performed locally on \(\Sys\). 
The classicality of the environmental memory can then be related to the separability of the process tensor~\cite{tarantoCharacterisingHierarchyMultitime2023}. 
While the process tensor is an elegant object from a theory point of view, its experimental determination is certainly challenging since it requires full multi-time statistics of the process. 
%Therefore, in this Letter we relax the assumption about what is known about the process. 
By contrast, the results of this Letter are based on the system dynamics alone, and, thus, are both conceptually and experimentally more easily accessible.

We should note that there is an interesting angle to our approach, relating it to the existence of physically measurable quantum trajectories.
%adaptive continuous measurement schemes, or quantum trajectories \cite{Wiseman}. 
We will explore these connections later, establishing non-Markovian master equations that allow for such a trajectory representation.

%We consider a discretized \emph{dynamics} \(\mathcal{D}\), which is a 
Formally, we define a \emph{dynamics} \(\mathcal{D}\) on \(\Sys\) to be a 
family of completely positive trace-preserving (CPT) maps \(\mathcal{D} = (\cpt_n)\) mapping the system state from the initial time \(t_0\) to time \(t_n\).
This definition covers every physically valid evolution where the system
and its environment are initially in a product state (uncorrelated).
%In the continuous case the parameter usually labels the time \(t\) and the family is often called the \emph{dynamical map} denoted by \(\cpt_t\).
To determine the dynamics, channel tomography has to be performed for each \(\cpt_n\), but no multi-time statistics is needed. 
Besides this experimental advantage of the approach, it conforms very well with the traditional open quantum system frameworks based on dynamical maps and master equations.
%which provides the dynamics of \(\Sys\) but not a process tensor. 

%Restricting ourselves to the dynamics of \(Sys\),  
In this letter, we show how to disclose a truly quantum memory for non-Markovian dynamics, based on such local information. 
The proposed witness thus locally reveals a new, additional property of quantum non-Markovian dynamics,
which is hidden for all known measures of non-Markovianity.
%that memory effects of a 
%non-Markovian dynamics.
%are due to a truly quantum environment and cannot be explained by classical memory. 
%This leads us to the central question we address in this Letter: \emph{If we have access only to the local system dynamics \(\mathcal{D}\), is it possible to verify that the memory effects in a non-Markovian dynamics are due to the quantum environment and cannot be explained by classical memory effects?}
%\todo{Erwähnung Quantentrajektorien bis hierher?}
%\medskip

\paragraph{Classical and Quantum Memory---}
Let us illustrate the idea with a simple toy model of a two-step dynamics \(\mathcal{D}\), given by the CPT maps
\begin{align}
	\cpt_1\left[\rho_\Sys\right] &= \tr_\Env\left[ U_1 ( \rho_\Sys \otimes \rho_\Env) U_1^\dagger\right], \notag \\
	\cpt_2\left[\rho_\Sys\right] &= \tr_\Env\left[ U_2 U_1 ( \rho_\Sys \otimes \rho_\Env) U_1^\dagger U_2^\dagger\right],
 \label{eq:simple-example-maps}
\end{align}
where \(\rho_\Sys\) and \(\rho_\Env\) are the initial states of 
system \(\Sys\) and environment \(\Env\), respectively.
The global dynamics are mediated by unitaries \(U_{1,2}\).
%(see Fig.~\ref{fig:CollisionToyModel}). 
%Locally, the maps describing the evolution of system \(\Sys\) alone are given by
%Any two-step dynamics can be written in this way if the environment and the unitaries are appropriately chosen.
%Now consider a case whereare qubits. 
We let both, \(\Sys\) and \(\Env\) be qubits, and set \(\rho_\Env = \kb{0}{0}\). Crucially, for our toy model we fix the second unitary to be the inverse of the first, i.e., \(U_2 = U_1^\dagger\).
Accordingly, the second CPT map is trivial, \(\cpt_2 = \id\).  
%The map in the first step is in general nontrivial and depends on the choice of \(U_1\).
As for \(U_1\), we consider two different choices:
\begin{align}
\label{eq:toy-model-interaction2}
	U_1^\text{dephase} &= \exp[ - i f (\sigma_x \otimes \sigma_x)], \notag\\
	U_1^\text{damp} &= \exp[- i g (\sigma_+\otimes \sigma_- + \sigma_-\otimes \sigma_+)],
\end{align}
with real parameters \(f\) and \(g\) determining the strength of the map~\footnote{See Supplemental Material at [ ... ] for details, which includes Refs.~\cite{breuerColloquiumNonMarkovianDynamics2016,liConceptsQuantumNonMarkovianity2018,rivasEntanglementNonMarkovianityQuantum2010,breuerMeasureDegreeNonMarkovian2009,lorenzoGeometricalCharacterizationNonMarkovianity2013a,smirneConnectionMicroscopicDescription2021,breuerTheoryOpenQuantum2007,kretschmerCollisionModelNonMarkovian2016,garrawayNonperturbativeDecayAtomic1997,diosiNonMarkovianQuantumState1998,heinekenQuantummemoryenhancedDissipativeEntanglement2021,bethruskaiAnalysisCompletelypositiveTracepreserving2002,zimanDescriptionQuantumDynamics2005,zimanOpenSystemDynamics2010a,hallCanonicalFormMaster2014,wisemanQuantumMeasurementControl2009,brunSimpleModelQuantum2002,diosiNonMarkovianContinuousQuantum2008,wisemanPureStateQuantumTrajectories2008,breuerGenuineQuantumTrajectories2004,pellegriniNonMarkovianQuantumRepeated2009,kronkeNonMarkovianQuantumTrajectories2012,smirneRateOperatorUnraveling2020,megierContinuousQuantumMeasurement2020,diosiNonMarkovianStochasticSchrodinger1997,breuerStochasticWavefunctionMethod1999,jackNonMarkovianQuantumTrajectory1999,piiloNonMarkovianQuantumJumps2008,hartmannExactOpenQuantum2017,gasbarriStochasticUnravelingsNonMarkovian2018,linkNonMarkovianQuantumDynamics2022a,beckerQuantumTrajectoriesTimelocal2023,karasikHowManyBits2011,daryanooshQuantumJumpsAre2014,beyerCollisionmodelApproachSteering2018}}.
The first leads to a partial dephasing in \(x\)-basis, the second choice induces a partial amplitude damping  (see Fig.~\ref{fig:bloch-spheres}).
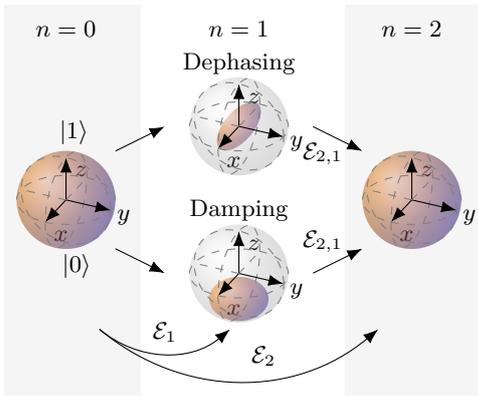
\begin{figure}
	\centering
\begin{tikzpicture}[scale=0.65]
    \begin{scope}[tdplot_rotated_coords]
        \filldraw [fill=lightgray, opacity=0.15, draw=none] (0, -4.75, -4) rectangle ++(0, 2.75,8);
        \filldraw [fill=lightgray, opacity=0.15, draw=none] (0, 4.9, 4) rectangle ++(0, -2.75,-8);
        \draw[-{Latex[length=2mm]}] (0, -2.5, 1) -- ++(0, 1, 0.5);
        \draw[-{Latex[length=2mm]}] (0, 1.5, 1.5) -- ++(0, 1, -0.5);
        \draw[-{Latex[length=2mm]}] (0, -2.5, -1) -- ++(0, 1, -0.5);
        \draw[-{Latex[length=2mm]}] (0, 1.5, -1.5) -- ++(0, 1, 0.5);
        \node (1) at (0, -3, -2.5) {};
        \node (2) at (0, 0, -2.5) {};
        \node (3) at (0, 3, -2.5) {};
        \draw[-{Latex[length=2mm]}] (1) to [bend right=40] (2);
        \draw[-{Latex[length=2mm]}] (1) to [bend right=40] (3);
        \node (ad) at (0, 0, 2.8) {Dephasing};
        \node (deph) at (0, 0, -0.2) {Damping};
        \node (e1) at (0, -1.5, -2.7) {$\cpt_1$};
        \node (e2) at (0, 0.5, -3.2) {$\cpt_2$};
        \node (e12) at (0,1.7, 0.95) {$\cpt_{2,1}$};
        \node (e12) at (0, 1.7, -0.85) {$\cpt_{2,1}$};
        \node (i) at (0, -3.5, 3.5) {$n=0$};
        \node (ii) at (0, 0, 3.5) {$n=1$};
        \node (iii) at (0, 3.5, 3.5) {$n=2$};
    \end{scope}
    \begin{scope}[xshift=-3.5cm, tdplot_main_coords]
        \blochdynwith{0}{0}{0}{1}{2}{25}{0}{90}
    \end{scope}
    \begin{scope}[yshift=1.5cm, tdplot_main_coords]
        \blochdyn{0}{0}{0}{0.3}{1}{0}{108}{58}
    \end{scope}
    \begin{scope}[yshift=-1.5cm, tdplot_main_coords]
        \blochdyn{0}{0}{-0.6}{0.6}{0.9}{25}{0}{90}
    \end{scope}
    \begin{scope}[xshift=3.5cm, tdplot_main_coords]
        \blochdyn{0}{0}{0}{1}{2}{25}{0}{90}
    \end{scope}
\end{tikzpicture}
	\caption{{Image of the Bloch sphere under the intermediate dephasing and damping dynamics, respectively.} Both dynamics are non-Markovian as witnessed by the expansion during the second step. Dephasing is realizable with only classical memory, while amplitude damping is not. {For this example, \(f = 0.64\) and \(g= 0.89\) in Eq.~\eqref{eq:toy-model-interaction2}. }}
 \label{fig:bloch-spheres}
\end{figure}
Almost any pure initial state of \(\Sys\) %\(\ket{\psi_\Sys}\), 
gets entangled with the environment in this first step, and therefore mixed. The second interaction then rewinds these correlations and the system returns to its initial state.
Thus, we witness non-Markovian dynamics according to all common criteria~\cite{Note1}. 
In this global picture it is fair to say that the repeated interaction with the same environmental quantum system
%system \(\Env\) induces the damping and dephasing, respectively, and that the second interaction with the same environment 
leads to non-Markovianity.
Clearly, \(\Env\) is that (quantum) memory.

However, once we look at the local dynamics \(\Dyn\) alone -- meaning that we know the maps \((\cpt_1, \cpt_2)\) but we are ignorant about the global dynamics including \(\Env\) -- the analysis is different:
any {single} qubit dephasing dynamics, {no matter what its true physical origin is, is indistinghuishable from a random unitary evolution} \cite{gregorattiQuantumLostFound2003,Helm2009}. {Since} classical memory suffices to keep track of the random choice of the unitary, no quantum environment \(\Env\) is needed~\cite{Note1}.
%to model the system evolution.
%Instead, the local dynamics of \(\Sys\) 
%can be described by 
%a random unitary process suffices, and thus  
%However, what if we look at the local dynamics \(\Dyn\) alone, meaning that we know the maps \((\cpt_1, \cpt_2)\) but we are ignorant about the global dynamics including \(\Env\)? Can we be sure that the latter served as a quantum memory which enabled the rewinding of the first step?
%In fact, for the dephasing dynamics it is well known that no environmental quantum system \(\Env\) is needed to produce the system evolution~\cite{GreWer2003}. Instead, the local dynamics in \(\Sys\) can be fully described by a random unitary process with classical memory~\cite{Note1}.

In the case of the partial amplitude damping, the situation is less obvious. 
%Clearly, this dynamics is not of random-unitary type and a quantum environment is necessary to obtain the non-unital evolution in the first step.
Remarkably, we will present a criterion below which verifies that the amplitude damping example indeed requires quantum memory, i.e., cannot be modeled by classical memory. 
%Nevertheless we may ask whether we really need a quantum \emph{memory} to generate the non-Markovian dynamics \(\Dyn\) or whether a classical one would suffice.
To proceed, we need to define properly what we mean by \emph{classical memory}:
%for two-step quantum dynamics: 
\begin{definition}
	\label{def:classicalmemory}
	Given two CPT maps $\cpt_1$ and $\cpt_2$. The dynamics $\Dyn = \dynamics{\cpt_1, \cpt_2}$ can be realized with \emph{classical memory}, iff there is at least one Kraus decomposition $\{M_i\}$ of $\cpt_1\left[\rho_\mathrm{S}\right] = \sum_{i}^{} M_i \rho_\mathrm{S} M_i^\dagger$ and suitable CPT maps $\Phi_i$ such that
	\begin{align}
		\label{eq:classicalmemory}
	\cpt_2\left[\rho_\mathrm{S}\right] = \sum_{i} \Phi_i \left[M_i \rho_\mathrm{S} M_i^\dagger\right].
	\end{align}
	Otherwise the dynamics is said to require \emph{truly quantum memory}.
\end{definition}

Let us elaborate why this definition embraces the idea of dynamics with \emph{classical memory}. 
Eq.~(\ref{eq:classicalmemory}) describes a sequential process. The Kraus decomposition \(\{M_i\}\) can be seen as a local measurement on \(\Sys\) which on average realizes the first map \(\cpt_1\).
%(possibly using a first quantum environment \(\Env_1\)).
The second step with CPT map \(\Phi_i\) is \emph{conditioned on that outcome} \(i\) of the first measurement. 
%(possibly realised with a second, independent quantum environment \(\Env_2\) ). 
Crucially, the label \(i\) is classical data, storable in a classical memory.
%-- by applying \(\Phi_i\) no (quantum) correlations with a possible \(\Env_1\) are upheld.
By contrast, for a dynamics that cannot be written in the form above, a persisting quantum environment has to be present throughout both dynamical steps, as suggested by our toy model

Further remarks: {The definition of a dynamics \(\Dyn\) requires that} \(\cpt_2\) is a CPT map from the initial time {\(t_0\)} to time \(t_2\).  {By contrast,} the average   map {from the intermediate time \(t_1\) to \(t_2\), given by} \(\cpt_{2,1} = \cpt_2 \circ \cpt_1^{-1}\), is in general \emph{not} CPT {(see also Fig.~\ref{fig:bloch-spheres})}. 
%It may not even exist and in particular it is in general not a mixture of the \(\Phi_i\). 
Moreover, for the actual implementation of the measurement \(\{M_i\}\) and the channels \(\Phi_i\), independent ancillary quantum systems
 {might be necessary. However,}  these can always be discarded after use, so they do not serve as a memory.

Markovian quantum dynamics satisfies Eq.~(\ref{eq:classicalmemory}) trivially with \(\Phi_i = \Phi = \cpt_{2,1}\), there is no memory at all. 
%In this case \(\cpt_{2,1}\) is CPT.
Any random unitary process (e.g., the dephasing in Fig.~\ref{fig:bloch-spheres}) can be written in the form of Eq.~(\ref{eq:classicalmemory}) of classical memory, as explained earlier.
% In the dephasing example above, the environmental ancilla \(\Env\) clearly serves as a quantum memory. However, its existence can only be verified if one has access to the global picture including the environment. 
% If we restrict the knowledge to the local dynamics itself, it  cannot be distinguished from the (essentially classical) random unitary realization. 
By contrast, the amplitude damping toy model cannot be realized in this way, as will follow from our theorem below.
%from the following general considerations:
%However, can the amplitude damping process from the example be written in this fashion as well? The answer is no. 

As the main result of this Letter, we next provide a sufficient criterion for a locally known dynamics \(\Dyn =\dynamics{\cpt_1,\cpt_2}\) to \textit{not} be realizable by means of classical memory according to Def.~\ref{def:classicalmemory}.
%This will be the central result of this Letter. 
%The criterion solely depends on knowledge about the local dynamics . 
Its relevance is twofold. 
First, if the criterion holds, 
%the memory effects of the dynamics cannot be explained classically. Thus, 
we have proof of a persistent quantum environment \(\Env\). %which repeatedly interacts with the system \(\Sys\). 
Second, note that Def.~\ref{def:classicalmemory} is the most general physically measurable pure-state quantum trajectory representation of the given dynamics \(\Dyn\).
%In fact, if the dynamics can be realized with classical memory, such a trajectory picture exists.
Disclosing quantum memory, therefore, rules out the existence of such quantum trajectories.
%the possibility of a physically measurable pure state quantum trajectory picture for the system dynamics.
We will elaborate on these issues in a time-continuous limit in more detail below. 
%We provide examples of dynamics that require truly quantum memory but also explicitly construct trajectory representations of non-Markovian time-local master equations with classical memory. 

%\bigskip
\paragraph{Criterion---}
For the criterion, we need the concept of \emph{entanglement of assistance}. 
Consider a bipartite quantum state $\choi_{\Sys\Anc}$ of system $\Sys$ and ancilla $\Anc$ (not to be confused with the environment $\Env$). Let \(E[\choi_{\Sys\Anc}]\) be an entanglement monotone (e.g., entanglement of formation or concurrence) \cite{plenioIntroductionEntanglementMeasures2007}. The entanglement of assistance \(\eoa\) is then \cite{divincenzoEntanglementAssistance1999,smolinEntanglementAssistanceMultipartite2005}
\begin{align}
	\label{eq:EoA}
	\eoa\left[\choi_{\Sys\Anc}\right] := \max_{\{p_k, \ket{\psi_k}\}} \sum_{k}^{} p_k E\left[\ket{\psi_k}\right],
\end{align}
i.e., the average entanglement 
%with respect to the entanglement monotone \(E\) 
\emph{maximized} over any pure-state decomposition of \(\choi_{\Sys \Anc}\). 
% In the original Ref.~\cite{DiVFucMabSmoThaUhl1999}, entanglement of formation was used as the defining entanglement monotone. 
% We will later use the \emph{concurrence of assistance} where the entanglement monotone in Eq.~(\ref{eq:EoA}) is chosen to be the concurrence.
%However, the following criterion itself is valid for any entanglement monotone \(E\).
Now assume that $\choi_{\Sys\Anc}$ describes the Choi state of the map $\cpt$ acting on the system $\Sys$~\cite{choiCompletelyPositiveLinear1975}, i.e.,
\begin{align}
    \choi_{\Sys\Anc}\left[\cpt\right] &= \choi\left[\cpt\right] = \left(\cpt \otimes \one\right)\ketbra{\phi^+}, \quad \nonumber \text{with} \\
    \ket{\phi^+} &=\frac{1}{\sqrt{d}}\sum_{j=0}^{d-1} \ket{j_\Sys}\ket{j_\Anc},
    \label{eq:bell-plus}
\end{align}
 {where \(d\) is the dimension of the system and the \(\ket{j_{\Sys,\Anc}}\) form an orthonormal basis in \(\Sys\) and \(\Anc\), respectively.}
We find the following theorem:
\begin{theorem}
	\label{th:centraltheorem}
	Let $\cpt_1$ and $\cpt_2$ be two CPT maps. If for the Choi states $\choi_1$ and $\choi_2$ of $\cpt_1$ and $\cpt_2$ we observe
	\begin{align}
		\eoa\left[\choi_1\right] < \eof\left[\choi_2\right],
	\end{align}
    the dynamics $\Dyn = \dynamics{\cpt_1, \cpt_2}$ requires quantum memory.
\end{theorem}

\begin{proof}
    Suppose the dynamics $\Dyn = \dynamics{\cpt_1, \cpt_2}$ only requires classical memory as defined in Def.~\ref{def:classicalmemory}. 
    Then the local measurement \(\{M_i\}\) implementing the channel \(\cpt_1\) on \(\Sys\) decomposes the  {corresponding} Choi state $\choi_1$ into the pure-state decomposition  $\{p_i, \ket{\psi_i}\}$ with $\ket{\psi_i} = \left(M_i \otimes \one \right)\ket{\phi^+}/\sqrt{p_i}$,  {\(\ket{\phi^+}\) as in Eq.~\eqref{eq:bell-plus}, and \(p_i\) beeing the probability for outcome \(i\).}
    The average entanglement in this decomposition  $\{p_i, \ket{\psi_i}\}$ is upper bounded by the entanglement of assistance:
	%We can now apply the definition of the entanglement of assistance in Eq.~\eqref{eq:EoA} to the Choi state $\choi_1$ of the map $\cpt_1$ and then compare it to an arbitrary decomposition $\choi_1 = \sum_{i} p_i \ketbra{\psi_i}$ 
	\begin{align}
	\label{eq:proofmaximisation}
	\eoa\left[\choi_1\right] =\max_{\{p_k, \ket{\psi_k}\}} \sum_{k}^{} p_k E\left[\ket{\psi_k}\right]\geq \sum_{i} p_i \eof\left[\ket{\psi_i}\right].
	\end{align}
	Local quantum channels can only reduce the entanglement \cite{plenioIntroductionEntanglementMeasures2007}. Therefore, defining $\rho_i \coloneqq \left(\Phi_i \otimes \one \right)\ketbra{\psi_i}$,  {where \(\Phi_i\) is a CPT map that can depend on the previous outcome,} we have
	\begin{align}
	\label{eq:proofLOCC}
	\sum_{i} p_i \eof\left[\ket{\psi_i}\right] &\geq 
    %\sum_{i}^{} p_i \eof \left[\left(\Phi_i \otimes \one \right)\ketbra{\psi_i}\right] = 
    \sum_{i}^{} p_i \eof \left[\rho_i\right].
	\end{align}
    The decomposition $\{p_i, \rho_i\}$ represents the Choi state $\choi_2$ of the second map $\cpt_2$,  {i.e., \(\sum_i p_i \rho_i = \choi_2\)}. 
	% such that there are local quantum channels $\Phi_i$ acting on part \(A\) yielding $\choi_2$ 
	% \begin{align}
	% \label{eq:proofChoiLocalTrajs}
	% \choi_2 = \sum_{i}^{} p_i \left(\Phi_i \otimes \one \right)\ketbra{\psi_i^{A B}}.
	% \end{align}
    However, the average entanglement in this decomposition is lower bounded by the entanglement of the state $\choi_2$ itself.
% But also $\choi_2$ admits different pure state decompositions, and we can now rewrite Eq.~\eqref{eq:proofLOCC2} in terms of the decomposition that minimizes the total entanglement. This satisfies, by definition of the minimum of a function, the inequality
	\begin{align}
	\label{eq:proof_final_eq}
	\sum_{i}^{}p_i \eof\left[\rho_i\right] &\geq \min_{\{p_k, \ket{\varphi_k}\}} \sum_{k} p_k \eof\left[\ket{\varphi_k}\right]
	=\eof\left[\choi_2\right],
	\end{align}
	where the minimization runs over all pure-state decompositions of $\choi_2$. 
 %Thus, combining Eqs.~\eqref{eq:proofmaximisation}--\eqref{eq:proof_final_eq} we have proved the theorem.
\end{proof}

%Since the entanglement of formation as well as the entanglement of assistance are computationally hard, for qubits we get an analogue statement for the concurrence of formation \cite{Wootters1998} and the concurrence of assistance. The concurrence of assistance is defined in terms of the concurrence of formation \cite{LauVerEnk2002}
% \begin{align}
% \label{eq:defCOA}
% \coa\left[\rho_{AB}\right] = \max_{\{p_k, \ket{\psi_k^{AB}}\}} \sum_{k} p_k \cof\left[\ket{\psi_k^{AB}}\right].
% \end{align}
% \begin{corollary}
% 	\label{cor:centralcorollary}
% 	Let $\cpt_1: \mathcal{S}\left(\Hilb\right) \to \mathcal{S}\left(\Hilb\right)$ and $\cpt_2: \mathcal{S}\left(\Hilb\right) \to \mathcal{S}\left(\Hilb\right)$ be two CPT qubit maps. The dynamics $\dynamics{\cpt_1, \cpt_2}$ has truly quantum memory, if for the Choi states $\choi_1$ and $\choi_2$ of $\cpt_1$ and $\cpt_2$ we find
% 	\begin{align}
% 	\label{eq:centralcorollary}
% 	\coa\left[\choi_1\right] < \cof\left[\choi_2\right].
% 	\end{align}
% \end{corollary}
% Due to the monotonic relation between the concurrences and the entanglements the proof is similar to the proof of Theorem~\ref{th:centraltheorem}.

\paragraph{Discrete example---}
First, we show a two-step dynamics that, upon changing a parameter, can be tuned from the case of verifiable quantum memory according to Thm.~\ref{th:centraltheorem} to the case of classical memory, obeying a representation as in Def.~\ref{def:classicalmemory}.  
We consider a map $\maptad{p}$ representing a thermal amplitude damping of a single qubit given by Kraus operators
%Second, we apply the criterion to the time-continuous non-Markovian amplitude damping of a qubit, which shows that the quantumness of the memory can only be  verified for sufficiently low temperatures.
%In the end we propose a way to explicitly construct non-Markovian master equations which only require classical memory, highlighting that quantum non-Markovianity is not necessarily connected to quantum memory effects.
%Depending on the choice of the relevant parameters we can continuously tune between different types of memory.
\begin{align}
\label{eq:discrete-damping-Kraus}
M_1 &= z_- \sqrt{p}\sigma_-,\!
&&M_2 = z_-(\sqrt{1-p}\sigma_+ \sigma_- + \sigma_-\sigma_+), \notag\\
M_3 &= z_+\sqrt{p}\sigma_+,\!
&&M_4 = z_+ (\sigma_+ \sigma_- + \sqrt{1-p} \sigma_-\sigma_+) ,
\end{align}
where the strength of the channel is given by $p\in \left[ 0, 1\right]$ and
\(z_\pm = 1/\sqrt{1+e^{\pm\beta}}\), 
%$(z_+/z_-)^2=\e^{-\beta}$ 
with $\beta$ a dimensionless inverse temperature.
%with \(z_\pm=\sqrt{({1\pm\temp)/{2}}} \). The parameter  \(\temp \in \left[-1, 0 \right]\) describes the temperature. 
The zero-temperature amplitude damping channel
%For $\temp=-1$ this corresponds to the amplitude damping channel 
with ground state $\ket{0}$ as its fixed point emerges as
$\beta \to \infty$.  At finite temperature, $M_3$ and $M_4$ model absorption from a thermal bath.

We consider a sequence of two maps of this class, i.e., a dynamics $\Dyn=\dynamics{\cpt_1,\cpt_2}=\dynamics{\maptad{p_1}, \maptad{p_2}}$, with \(p_n\) the damping strength at time \(t_n\). For the sake of this example, we fix the inverse temperature  $\beta=0.51$,
%\(\temp=-0.25\) 
the first damping strength \(p_1 = 0.9\), and investigate the nature of the required memory as a function of the second strength \(p_2\). 
We choose the concurrence \(\cof\) as the entanglement monotone \(E\) in Eq.~\eqref{eq:EoA} and write \(\coa\) for the concurrence of assistance. 
In Fig.~\ref{fig:tunable_quantum_memory} we plot \(
    \coa[\chi_1] - \cof[\chi_2] 
%= \coa_1 - \cof_2 
\)
and satisfy the criterion for $p_2<0.11$ (orange region). Thus, the corresponding non-Markovian dynamics requires quantum memory. %Clearly, for $p_2<0.11$ (orange) our criterion is satisfied and the non-Markovian dynamics requires quantum memory.
%First of all, for $p_1 \leq p_2$ (the damping at time \(t_2\) is stronger than at \(t_1\)) the dynamics is Markovian and thus does not require memory at all (hatched region). 
For \(p_2 > 0.86\), the dynamics can be modeled by classical memory (blue region). We provide an explicit representation as in Def.~\ref{def:classicalmemory} (see caption of Fig.~\ref{fig:tunable_quantum_memory} for details).

\begin{figure}
    \centering
    \includegraphics[width=\columnwidth]{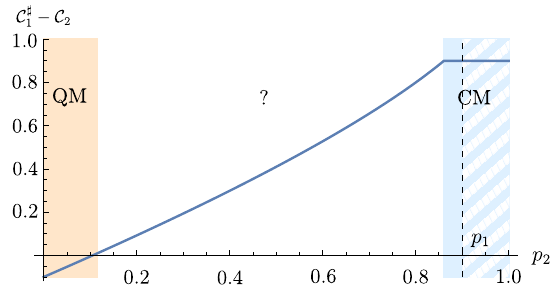}
    \caption{Entanglement difference as a function of the strength parameter $p_2$ of the second dynamical step of a thermal amplitude damping channel
    (see text, other parameters $p_1=0.9$ and \(\beta=0.51\)). For \(p_2 < 0.11\), the criterion in Thm.~\ref{th:centraltheorem} is satisfied (orange) and the dynamics requires quantum memory (QM). 
    For $p_2 < p_1$, the damping gets partially rewound and the dynamics is non-Markovian. Yet for \(0.86 \leq p_2 < p_1\) (solid blue) we can explicitly construct a representation as in 
    %measurement \(\{M_i\}\) and channels \(\Phi_i\) as in 
Def.~\ref{def:classicalmemory}, and therefore only classical memory (CM) is needed -- see Supplemental Material~\cite{Note1}).
    For $p_2 \geq p_1$ the dynamics is Markovian and thus does not require memory at all (blue hatched region).
    For \(0.11 \leq p_2 < 0.86\) (white) we cannot decide whether truly quantum memory is required.}
\label{fig:tunable_quantum_memory}
\end{figure}

\paragraph{Time-continuous example---}

Let us apply the criterion of Thm.~\ref{th:centraltheorem} to the zero-temperature non-Markovian amplitude damping master equation,
%We consider the non-Markovian damping of a two-level system \todo{physikalischer Hintergrund, cite some}.
%The zero-temperature master equation takes the form \todo{zero-temperature sufficient?}
\begin{align}
\label{eq:nMadthLindblad}
\dot \rho=\mathcal{L}_t\left[\rho\right] = \frac{\gamma_-(t)}{2} \left(\left[\sigma_- \rho, \sigma_+\right] + \left[\sigma_- ,\rho  \sigma_+\right]\right).
%+\frac{\gamma_+(t)}{2} \left(\left[\sigma_+ \rho, \sigma_-\right]  + \comm{\sigma_+ \rho,  \sigma_-}\right),
\end{align}
Here \(\gamma_-(t)\) is the instantaneous damping rate which in the non-Markovian case changes sign over time~\cite{Note1}.

%\todo{how to describe what $\coupling$ is?, how to treat long equations here?}

For the dynamics resulting from this master equation, we find that 
the concurrence of assistance of the Choi state is equal to the concurrence (of formation) for all times,
%\begin{align}
$\coa\left[\choi(t)\right]=\cof\left[\choi(t)\right], \forall t$.
%\end{align}
In the non-Markovian case, \(\cof\) is a non-monotonous function. Thus, there are times \(t_2 > t_1\) such that
%\begin{align}
 $\coa\left[\choi(t_1)\right] < \cof\left[\choi(t_2)\right]$,
%\end{align}
which shows by virtue of Thm.~\ref{th:centraltheorem} that zero-temperature non-Markovian amplitude damping cannot be realized by means of classical memory.

However, heuristically extending the scenario to a thermal bath, one finds that for sufficiently high temperatures the criterion is no longer violated~\cite{Note1}. 
This does not necessarily mean that the dynamics can be explained by classical memory, but it shows that at higher temperatures it becomes harder to locally verify the quantum nature of the memory.
%in the environment. 

\paragraph{Dynamics with classical memory---}
Disclosing quantum memory for time-continuous dynamics requires the consideration of the dynamical map at two distinct times, as seen in the previous example. However, to ensure that classical memory is sufficient, one has to explicitly provide a representation in terms of the time-continuous generalization of Def.~\ref{def:classicalmemory}.

For a dynamics \(\Dyn=\left(\cpt_n\right)_{n=1}^{N}\), with \(N\) discrete time steps,
a representation with classical memory takes the form
\begin{align}
		\label{eq:time-continuous-classical}
	\cpt_n\left[\rho\right] = \sum_{i_1, \ldots, i_n}  &M_{i_n}^{(i_1,\ldots, i_{n-1})} \ldots M_{i_2}^{(i_1)} M_{i_1}\rho M_{i_1}^\dagger M_{i_2}^{(i_1)\dagger} \ldots \notag \\
 &\ldots M_{i_n}^{(i_1,\ldots, i_{n-1})\dagger}, \quad 1\le n \le N,
\end{align}
where the superscripts indicate that the measurement operators at a certain step can depend on \emph{all} previous outcomes.
For suitably chosen measurements \(\{M_{i_n}\}\), this construction allows for a time-continuous limit.

Eq.~\eqref{eq:time-continuous-classical} describes the most general form of a physically measurable pure-state trajectory representation of a dynamics.
Hence, for a dynamics which requires truly quantum memory according to Thm.~\ref{th:centraltheorem}, a pure-state unraveling is immediately ruled out. 
On the other hand, a non-Markovian dynamics which can be written in this way, i.e., which only requires classical memory, admits a pure-state trajectory representation by construction. This clarifies that the often debated existence of physically measurable non-Markovian quantum trajectories depends on the classicality of the memory needed to implement the dynamics~\cite{Note1}. In the following, we provide some time-continuous examples.

As mentioned earlier, any dynamics with random unitary representation can be realized with classical memory  {(see also Refs.~\cite{grotzQuantumDynamicsFluctuating2006a,chruscinskiNonMarkovianRandomUnitary2013b})}.
Another simple case is a probabilistic mixture of multiple Markovian dynamics. 
The prime example is the master equation of eternal non-Markovianity requiring two bits of classical memory~\cite{megierEternalNonMarkovianityRandom2017}. 
There, the outcome of an initial random choice with probabilities $p_i$ determines which of three different Markovian dynamics with generators $\mathcal{L}_i$ is implemented for all times,
%This results in a mixture
%\begin{align}
%\label{eq:MixingDynamics}
$\cpt_t = p_1 \e^{t \mathcal{L}_1} + p_2 \e^{t \mathcal{L}_2} + p_3 \e^{t \mathcal{L}_3}$.
%\end{align}
This dynamical map is in general non-Markovian with respect to the CP-divisibility criterion~\cite{megierEternalNonMarkovianityRandom2017,filippovDivisibilityQuantumDynamical2017, breuerMixinginducedQuantumNonMarkovianity2018,jagadishConvexCombinationsPauli2020}. Nevertheless, it has an obvious pure-state trajectory representation. 
Further dynamics with classical memory are given by quantum semi-Markov processes, where the application of the next step depends on a (classical) waiting time distribution~\cite{breuerQuantumSemiMarkovProcesses2008,vacchiniGeneralizedMasterEquations2016,chruscinskiSufficientConditionsMemorykernel2016,megierEvolutionEquationsQuantum2020,megierMemoryEffectsQuantum2021}.

% All outcomes $i$ are sampled from a classical probability distribution and the corresponding measurement operators \(\{M_{i_n}\}\) are proportional to reversible unitaries. \todo{link Gregoratti Werner supplement}

% Another type of dynamics which always only requires classical memory is random unitary dynamics. Suppose $\cpt_1$ is random unitary. Then any measurement with a measurement operator proportional to a unitary operator will not reveal any information about the system but in the same way also it will not disturb it. This in turns means that all successively applied operations can be undone step-wise, returning to the initial state.

The richness of dynamics with classical memory is, however, far greater. Eq.~\eqref{eq:time-continuous-classical} can serve as a starting point to derive new non-Markovian master equations with classical memory based on a quantum-jump-inspired trajectory representation, as we show next.

% As a concluding example we propose a construction of non-Markovian master equations with classical memory that is based on a quantum-jump-inspired trajectory representation.

%Consider two measurement schemes with two measurement operators each $\mathcal{M_1} = \{M_1, M_2\}$ and $\mathcal{M_2} = \{M_1, M_2\}$.

We use a qubit and start from a standard quantum jump trajectory which describes amplitude damping (jump operator $\sigma_-$).
The classical memory keeps track of whether the jump has already occurred.
If so, the jump operator is replaced by $\sigma_+$.
%In the simplest qubit case one starts from a standard quantum jump trajectory approach which describes an amplitude damping and keeps track of whether the jump has already occurred. 
%As soon as a jump is observed, the damping is reversed. 
One bit of classical memory is sufficient for the implementation of this scheme. 
Integrating the succession of maps over all possible jump times yields the non-Markovian time-local master equation
\begin{align}
\label{eq:basicmodellindblad}
\mathcal{L}_t\left[\rho\right] = \frac12 \sum_{k=1, 2} \gamma_k(t)  \left(\left[L_k, \rho L_k^\dagger \right]+\left[L_k\rho, L_k^\dagger\right] \right),
\end{align}
with
% \begin{align}
% \gamma_1(t) &= \frac{\damp  \left(\e^{\damp  t}-1\right)}{4 \left(\e^{\damp  t}-\damp
% 	t\right)},
% &&\hphantom{L_2}\gamma_2(t)=\frac{\damp  (\damp  t-1)}{2 (\damp  t- \e^{\damp  t})},\\
% L_1 &= \frac{\sigma_z}{\sqrt{2}},
% &&\hphantom{\gamma_2(t)}L_2 = \sigma_-.
% \end{align}
\begin{align*}
\gamma_1(t)&=\frac{\damp  (\damp  t-1)}{2 (\damp  t- \e^{\damp  t})},
&&\hphantom{L_2}\gamma_2(t) = \frac{\damp  \left(\e^{\damp  t}-1\right)}{8 \left(\e^{\damp  t}-\damp
	t\right)},\\
L_1 &= \sigma_-,
&&\hphantom{\gamma_2(t)}L_2 = \sigma_z.
\end{align*}
A detailed derivation is presented in the Supplemental Material~\cite{Note1}.
Let us stress that the non-Markovian master equation~\eqref{eq:basicmodellindblad}  has a physically realizable, measurable quantum jump representation by construction. 
%A similar construction can be found in Ref.~\cite{smirneRateOperatorUnraveling2020}.

It is interesting to note that the dynamics given by Eq.~\eqref{eq:basicmodellindblad} is P-indivisible and thus non-Markovian in a stricter sense than the P-divisible master equation of eternal non-Markovianity discussed earlier. %Eq.~\eqref{eq:MixingDynamics}.

\begin{figure}
\begin{center}
\begin{tikzpicture}
    \draw[rounded corners=7pt] (0.5, 1.5) rectangle (4.5, 3.5);
    \draw[fill=mathematicablue, fill opacity=0.8, rounded corners=7pt]   (0.5,1.5) rectangle ++(4,2);
    %\draw[pattern={mylines[size= 5pt,line width=0.8pt,angle=45]},  pattern color=black]   (0.5,3) rectangle (4.5,3.5);
    \fill [pattern={mylines[size= 5pt,line width=0.8pt,angle=45]},  pattern color=mathematicadarkblue]
      (0.5,3) --
      ++(4,0) {[rounded corners=7] --
      ++(0,0.5) --
      ++(-4,0)} --
      cycle
      {};
    \draw[draw, color=mathematicadarkblue] 
      (0.5,3) --
      ++(4,0) {[rounded corners=7] --
      ++(0,0.5) --
      ++(-4,0)} --
      cycle
      {};
    %\draw[fill=white] (0.5, 1.5) rectangle (3, 3);
    %\draw[fill=mathematicaorange, fill opacity=0.8] (0.5, 1.5) rectangle (3, 3);
    %\draw[pattern={mylines[size= 5pt,line width=5pt,angle=135]},  pattern color=white] (0.5, 1.5) rectangle (3, 3);
    \fill[fill=white, fill opacity=1]  
      (3,3) --
      ++(0,-1.5) {[rounded corners=7] --
      ++(-2.5, 0)} --
      ++(0, 1.5) --
      cycle
      {};
    \fill[fill=mathematicaorange, fill opacity=0.8]  
      (3,3) --
      ++(0,-1.5) {[rounded corners=7] --
      ++(-2.5, 0)} --
      ++(0, 1.5) --
      cycle
      {};
    %\fill[pattern={mylines[size= 5pt,line width=5pt,angle=135]},  pattern color=white] 
      %(3,3) --
      %++(0,-1.5) {[rounded %corners=7] --
      %++(-2.5, 0)} --
      %++(0, 1.5) --
      %cycle
      %{};
    \draw[color=mathematicadarkblue] (3, 1.5) -- ++(0, 1.5);
    \draw[rounded corners=7pt, line width=1.5pt, color=mathematicadarkblue] (0.5, 1.5) rectangle (4.5, 3.5);
    \draw[fill=white, rounded corners=7pt] (0.6, 1.7) rectangle (2.9, 2.25) node[pos=.5] {};
    \draw[fill=mathematicaorange, fill opacity=0.8, rounded corners=7pt, line width=0pt] (0.6, 1.7) rectangle (2.9, 2.25) node[pos=.5] {};
    \draw[rounded corners=7pt, color=mathematicadarkorange, line width=1.1pt, opacity=0.8] (0.6, 1.7) rectangle (2.9, 2.25) node[pos=.5, color=black, opacity=1] {\small{$\eoa(t_1) < \eof(t_2)$}};
    \fill[color=white, fill opacity=0.85] (1.5, 3.05) rectangle (3.5, 3.45) node[pos=.5, color=black, opacity=1] {Markovian};
    \fill[color=white, fill opacity=0.85] (1.1, 2.45) rectangle (3.9, 2.85) node[pos=.5, color=black] {non-Markovian};
    \draw[fill=mathematicablue, fill opacity=0.8]   (5,2.75) rectangle (5.75,3.25);
    \draw[color=white] (6.5, 3) rectangle (8, 3) node[pos=0.5, color=black] {classical memory};
    \draw[fill=mathematicaorange, fill opacity=0.8]   (5,1.75) rectangle (5.75,2.25);
    \draw[color=white] (6.5, 2) rectangle (8, 2) node[pos=0.5, color=black] {quantum memory};
\end{tikzpicture}
\end{center}
    \caption{Non-Markovian quantum dynamics may emerge from classical or truly quantum memory. Markovian dynamics trivially falls into the first class.
    %Markovian dynamics (memoryless) requires classical memory only, as does 
    Some dynamics requiring a truly quantum memory can be detected by the criterion proposed in Thm.~\ref{th:centraltheorem}  {(encircled region within the orange area). An example is the non-Markovian amplitude damping (see Eq.~\eqref{eq:nMadthLindblad}). Examples of non-Markovian dynamics with classical memory (solid blue region) encompass Eq.~\eqref{eq:basicmodellindblad} and the master eqation of enternal non-Markovianity~\cite{megierEternalNonMarkovianityRandom2017}. We conjecture that there exist dynamics which require quantum memory (orange region) but which cannot be detected by our criterion.}}
    \label{fig:venn_like}
\end{figure}
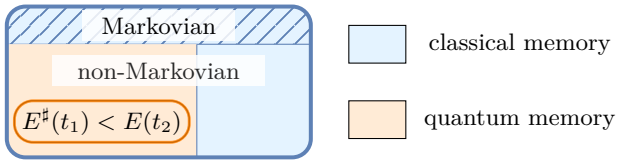

\paragraph{Conclusion---}
Non-Markovian quantum dynamics is associated with memory effects. However, this memory is not necessarily provided by environmental quantum degrees of freedom but may be classical. In this Letter, we investigate the nature of that memory from a local viewpoint. Focusing on the dynamics in the open system alone, we make no assumption about the physics of the environment.

We start from a definition for a dynamics requiring classical memory only.
As the main result, we then present a criterion in terms of an inequality whose satisfaction rules out any such realization of the given dynamics.
Crucially, this criterion depends solely on information about the single-time local dynamics of the open system, no multi-time statistics is required.
Its tomography and thus the disclosure of environmental quantum memory is in experimental sight.
%easier than a full characterization of the global dynamics including the environment.
%This renders our approach an experimentally applicable tool for the verification of quantum memory effects in unknown environments.

We illustrate the concept with several discrete and time-continuous 
%dynamics providing 
examples with and without truly quantum memory, including cases which can be tuned between the two regimes.
In particular, we show how to construct a class of non-Markovian time-local master equations admitting a pure-state quantum jump trajectory representation based on classical memory. No such unraveling can exist for 
a dynamics which requires truly quantum memory.
%cannot have such a trajectory picture. 

Our criterion is sufficient but not necessary -- refinements are thus desirable (see Fig.~\ref{fig:venn_like}). The presented concepts serve as an immediate starting point for further investigations, which include characterizing the size of the quantum or classical memory, criteria for unital dynamics, and the construction of physically realizable non-Markovian trajectories of the diffusive type.
More generally, {our work shows that an environment-agnostic perspective can be a valuable tool for characterizing environmental properties without making prior assumptions about the underlying physics.}

\section{Acknowledgements}
We are grateful for fruitful discussions with Dario Egloff, Kimmo Luoma, and Andrea Smirne. Further, we acknowledge helpful comments by many participants at the SMP54 conference in Toru\'n, Poland.

\bibliography{literature_memory}

\clearpage

\title{Supplemental Material: Local disclosure of quantum memory in non-Markovian dynamics}

\maketitle

\setcounter{page}{1}
\renewcommand{\theequation}{S.\arabic{equation}}
\setcounter{equation}{0}

\section{Definitions of quantum non-Markovianity}
Classical non-Markovianity is a well defined term. 
For quantum dynamics, however, there are several nonequivalent but closely related definitions based on different motivations. 
Here, we will list some important instances. 
Detailed reviews of quantum non-Markovianity can be found in Refs.~\cite{breuerColloquiumNonMarkovianDynamics2016,liConceptsQuantumNonMarkovianity2018}.

\paragraph{Divisibility:}
A dynamics is called CP-divisible, whenever for any choice of two successive maps $\cpt_1, \cpt_2$ the intermediate map $\cpt_{2,1}= \cpt_2 \circ \cpt_1^{-1}$ is a completely positive map (CP-map).
If \(\cpt_{2,1}\) is a positive but not completely positive map, the dynamics is called P-divisible.  
If \(\cpt_{2,1}\) is not positive, the dynamics is said to be indivisible~\cite{rivasEntanglementNonMarkovianityQuantum2010}.

\paragraph{Trace distance:}
The trace distance is an indicator for the distinguishability of two quantum states. 
Under Markovian dynamics, the trace distance of any two states can only decrease over time. Thus, an increase of the trace distance between two state is an indicator for non-Markovianity~\cite{breuerMeasureDegreeNonMarkovian2009}.

\paragraph{Entanglement with an ancilla:}
Let us assume the system of interest \(\Sys\) is initially entangled with an ancilla \(\Anc\) (which is not to be confused with the environment). Under the application of a Markovian dynamics to the system \(\Sys\) alone, the entanglement between \(\Sys\) and \(\Anc\) can only decrease. 
Thus, an increase of the entanglement is a witness of non-Markovianity in the dynamics of \(\Sys\)~\cite{rivasEntanglementNonMarkovianityQuantum2010}.

\paragraph{State space volume:}
The volume of states that can be reached by a dynamics can only shrink under Markovian evolution.
Therefore, similar to the trace distance criterion, non-Markovianity can be detected if the state space volume increases over time~\cite{lorenzoGeometricalCharacterizationNonMarkovianity2013a}.

\section{Qubit-Qubit Toy Model}

Starting from the interactions of the form in Eq.~(2) between system and environment, the local dynamics on the system $\Sys$ can be derived by tracing out the environment.
Initializing the environment $\Env$ in the $\ket{0}$-state and setting
\begin{align}
    f = \frac{1}{2} \arccos{(1-p)}, \quad 
    g = \arcsin{(\sqrt{p})},
\end{align}
with $p\in\left[0, 1\right]$ defining the strength of the process, we obtain the familiar amplitude damping and dephasing channel, respectively.
A Kraus representation of the dephasing channel $\cpt_{\mathrm{dephase}}$ takes the form
\begin{align}
K_1 = \sqrt{p/2} \sigma_x,
% \begin{pmatrix}
%   0 & \sqrt{p/2} \\
%  \sqrt{p/2} & 0
%  \end{pmatrix},
&& K_2 = \sqrt{1-p/2} \id,
% \begin{pmatrix}
%  \sqrt{1-p/2} & 0 \\
%  0 & \sqrt{1-p/2}
%  \end{pmatrix},
\end{align}
while the amplitude damping channel $\cpt_{\mathrm{damp}}$ is given by
\begin{align}
K_1 = \sigma_-
% \begin{pmatrix}
%  0 & 0 \\
%  \sqrt{p} & 0
%  \end{pmatrix},
 && K_2 = \sqrt{1-p} \sigma_+\sigma_- + \sigma_-\sigma_+.
 % \begin{pmatrix}
 % \sqrt{1-p} & 0 \\
 % 0 & 1
 % \end{pmatrix}.
\end{align}

The dynamics $\Dyn_\textrm{dephase} = (\cpt_{\mathrm{dephase}}, \id)$ as well as the dynamics $\Dyn_\textrm{damp} = (\cpt_{\mathrm{damp}}, \id)$ are non-Markovian with respect to the different notions of non-Markovianity listed above. Both dynamics are indivisible for any nontrivial strength \(p\). 
Considering Fig.~1, we can directly see that the distance between $\ket{0}$ and $\ket{1}$ corresponds to the full diameter of the Bloch sphere for $n=0$ and $n=2$. In between these two steps, the images of $\ket{0}$ and $\ket{1}$ approach each other, which can also be seen in that figure. 
Equivalently, the increase of the state space from \(n=1\) to \(n=2\) is directly visible. 
The entanglement criterion is also fulfilled, as can be verified for example by starting with a Bell state of system \(\Sys\) and ancilla \(\Anc\).

\begin{figure}
	\begin{center}
		\begin{tikzpicture}
		\draw (-0.5, -0.5) node{$\rho_\Sys$};
		\draw (-0.5, -1.3) node{$\rho_\Env$};
		\qwire{0}{-0.5}{0.5}
		\qwire{0}{-1.3}{0.5}
		\twoqubgate{0.5}{0}{$U_1$}{0.6}
		\qwire{1.1}{-0.5}{1}
		\qwire{1.1}{-1.3}{1}
		\twoqubgate{2.1}{0}{$U_2$}{0.6}
		\qwire{2.7}{-0.5}{0.8}
		\qwire{2.7}{-1.3}{0.8}
        \draw[dashed] (1.6, -0.7)--++(0, 0.55);
        \node (1) at (1.6, -0.9) {$t_1$};
        \draw[dashed] (3.1, -0.7)--++(0, 0.55);
        \node (1) at (3.1, -0.9) {$t_2$};
		\end{tikzpicture}
	\end{center}
	\caption{A circuit of a two-step dynamics. The system \(\Sys\) interacts twice with the same environment \(\Env\). The local maps on \(\Sys\) for times \(t_{1,2}\) are given in Eq.~(1).}
	\label{fig:CollisionToyModel}
\end{figure}
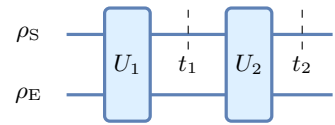

%\section{Entanglement is not the thing}
%\cite{smirneConnectionMicroscopicDescription2021}
% In order to show that the type of the memory has nothing to do with the question of whether there are quantum correlations with an environment, we will now demonstrate how entanglement behaves in the toy model. In both cases the initial $n=0$ entanglement with the environment, measured in terms of concurrence $\mathcal{C}$, is zero, as well as the final $n=2$ entanglement.
% For $n=1$ in both dynamics the entanglement with the environment increases, choosing $p=1/2$ for example leads to maximal entanglement \todo{Is that correct ???}.

\section{Random unitary representation of the dephasing dynamics}
Here, we show how the non-Markovian dephasing dynamics in Fig.~1 can emerge from random unitaries and classical memory.
The first map is realized by flipping a weighted coin, which shows heads with probability \(p/2\) and tails with probability \(1-p/2\). Whenever the coin shows heads, the unitary \(U_{\mathrm{heads}} = \sigma_x\) is implemented. In the case of tails, the trivial unitary \(U_\mathrm{tails} = \id\) is applied. The result of the coin toss is stored in a memory. 

In the second step, the dynamics can be rewound completely by applying the inverse of the unitary operation in the first step. Thus, if the state was flipped by \(U_\mathrm{heads}\), we apply \(U_\mathrm{heads}^\dagger = \sigma_x\) and the flipping is reversed. Trivially, for tails in the first step, nothing is done in the second step either.
Thus, the local dephasing dynamics can be realized by a random unitary process.
Crucially, this procedure can only be applied if the (classical) outcome of the initial coin toss has been stored. 

On the global level of \(\Sys\) and \(\Env\) the random unitary process just described is of course not equivalent to the circuit in Fig.~\ref{fig:CollisionToyModel}. However, on the local level of \(\Sys\) alone the two different realizations of the dynamics cannot be distinguished (see also Ref.~\cite{smirneConnectionMicroscopicDescription2021}).
The construction shown here can be extended to arbitrary random unitary processes.

\section{Regime with classical memory}
For the solid blue region in Fig.~2 of the main text, we can provide an explicit construction of a realization with classical memory as in Def.~1. Crucially, the dynamics is non-Markovian in this regime since the damping at time \(t_1\) is stronger than at \(t_2\), i.e., the evolution is partially rewound.

The map describing the evolution \(\cpt_1\) up to time \(t_1\) is given by the Kraus representation in Eq.~(10). 
A measurement of these Kraus operators would however not help to construct a realization with classical memory. Instead, we have to measure another set of Kraus operators defined by
\begin{align}
M_\alpha &= \frac{1}{\sqrt{2}}\left(M_1 + M_3\right), &&M_\beta=\frac{1}{\sqrt{2}}\left(M_1 - M_3\right),\\
M_\gamma &= M_2, &&M_\delta = M_4.
\end{align}
This Kraus representation equivalently implements the map \(\cpt_1\).

Depending on the outcome of this first measurement, different channels \(\Phi_i\) are implemented (cf.~Def.~1).
For the four different outcomes, they are given by:
%Fixing a temperature $\temp$, we can perform a damping with damping parameter $p_1$ using these four operators. The measurements in the second step conditioned on the first step are all unitary, they are given by
\begin{align}
	\alpha &\longrightarrow \Phi_\alpha\left[\placeholder \right] = \cpt_0[\sigma_x\placeholder \sigma_x],\\
	\beta &\longrightarrow \Phi_\beta\left[\placeholder \right] = \cpt_0[\sigma_y\placeholder \sigma_y],\\
	\gamma &\longrightarrow \Phi_\gamma\left[\placeholder \right] = \cpt_0[\placeholder],\\
	\delta &\longrightarrow \Phi_\delta\left[\placeholder \right] = \cpt_0[\placeholder],
\end{align}
where \(\cpt_0\) is a channel which is the same for each outcome and which must not be confused with the resulting map \(\cpt_2\). 
The Choi state of \(\cpt_0\) is given by
\begin{widetext}
\begin{align}
    \choi(\cpt_0)=\begin{pmatrix}
        \frac{1}{2} p ({z_0}-1)+1 & 0 & 0 & \frac{2 \sqrt{1-p}}{{p_1}
      \sqrt{1-{z_0}^2}+2 \sqrt{1-{p_1}}} \\
 0 & \frac{1}{2} p ({z_0}+1) & 0 & 0 \\
 0 & 0 & -\frac{1}{2} p ({z_0}-1) & 0 \\
 \frac{2 \sqrt{1-p}}{ {p_1} \sqrt{1-z_0^2}+2 \sqrt{1-p_1}  } & 0
   & 0 & 1-\frac{1}{2} p ({z_0}+1)
    \end{pmatrix},
\end{align}
\end{widetext}
where \(z_0 = -\tanh{\beta/2}\).
% AND THIS ONE HERE IS WRONG
% \begin{align}
% \choi(\cpt_0) = 
% \begin{pmatrix}
% \frac{1}{2} p_2 (\temp-1)+1 & 0 & 0 & \frac{2 \sqrt{1-p_2}}{p_1
% 	\sqrt{1-\temp^2}+2 \sqrt{1-p_1}} \\
% 0 & \frac{1}{2} p_2 (1+\temp) & 0 & 0 \\
% 0 & 0 & \frac{1}{2} p_2 (1-\temp) & 0 \\
% \frac{2 \sqrt{1-p_2}}{p_1 \sqrt{1-\temp^2}+2 \sqrt{1-p_1}} 
% & 0 & 0 & 1-\frac{1}{2} p_2 (\temp+1) \\
% \end{pmatrix}.
% \end{align}

On average this measurement procedure realizes the channel \(\cpt_2\), i.e., 
\begin{align}
    \cpt_2[\placeholder] = \sum_{i=\alpha,\beta,\gamma,\delta} \Phi_i[M_i \placeholder M_i^\dagger].
\end{align}
Note that this construction is only valid for the parameters corresponding to the blue region in Fig.~2. 
For other parameter choices, the map \(\cpt_0\) is no longer a quantum channel.

\section{Non-Markovian amplitude damping}
The non-Markovian amplitude damping master equation for a vacuum bath can be derived in several ways (see for example~Refs.~\cite{breuerTheoryOpenQuantum2007,kretschmerCollisionModelNonMarkovian2016,garrawayNonperturbativeDecayAtomic1997,diosiNonMarkovianQuantumState1998}).
The instantaneous damping rate \(\gamma_-\) depends on time and becomes negative during the evolution if the dynamics is non-Markovian. 
A possible parametrization of the time-dependent decay rate can be explicitly expressed as
% \begin{align}
%     \gamma_-(t) = -\frac{2 \coupling^2 \left(-\gamma +\sqrt{\gamma ^2-16 \coupling^2} \sinh
%    \left(\frac{1}{2} t \sqrt{\gamma ^2-16 \coupling^2}\right)+\gamma  \cosh
%    \left(\frac{1}{2} t \sqrt{\gamma ^2-16 \coupling^2}\right)\right)}{-\gamma
%     \sqrt{\gamma ^2-16 \coupling^2} \sinh \left(\frac{1}{2} t \sqrt{\gamma
%    ^2-16 \coupling^2}\right)+\left(8 \coupling^2-\gamma ^2\right) \cosh
%    \left(\frac{1}{2} t \sqrt{\gamma ^2-16 \coupling^2}\right)+8 \coupling^2},
% \end{align}
\begin{align}
    \gamma_-(t) = -\frac{2 \coupling^2 \left(-\gamma_0 +2\gammashort \sinh
 \gammashort t +\gamma_0  \cosh
   \gammashort t \right)}{-2\gamma_0
    \gammashort \sinh  \gammashort t+\left(8 \coupling^2-\gamma_0 ^2\right) \cosh
   \gammashort t +8 \coupling^2},
\end{align}
with
\begin{align}
    \gammashort = \frac{1}{2} \sqrt{\gamma_0 ^2-16 \coupling^2},
\end{align}
\(\gamma_0\) a constant overall damping strength, and \(\alpha\) a parameter which tunes the non-Markovianity of the process.

In Fig.~\ref{fig:CC-NM-Amp-Damp} we plot the concurrence of assistance \(\coa\) and the concurrence of formation \(\cof\) over time for parameters \(\gamma_0 = \alpha =1 \).
Both quantities  are equal at all times in the case of non-Markovian amplitude damping, i.e., \(\coa(t) = \cof(t)\).
The revivals of the entanglement are clearly visible and, thus, there are times \(t_1\) and \(t_2\) such that \(\coa(t_1) < \cof(t_2)\). Accordingly, the dynamics cannot be realized with classical memory only. It should be noted, though, that in the beginning of the dynamics there is a time interval for which the criterion cannot be satisfied. This is in agreement with the fact that the dynamics is divisible up to the first minimum in Fig.~\ref{fig:CC-NM-Amp-Damp} (\(t\approx 1.9\)) and thus a memoryless realization exists in this regime.
\begin{figure}
    \centering
    \includegraphics[width = \columnwidth]{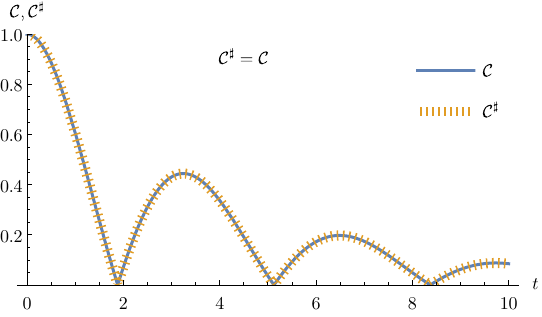}
    \caption{Concurrence of assistance \(\coa\) and concurrence of formation \(\cof\) of the Choi state of the dynamical map of the non-Markovian amplitude damping. Due to the revivals, there are times \(t_1\) and \(t_2\) such that \(\coa(t_1) < \cof(t_2)\) and the dynamics verifiably requires quantum memory.}
    \label{fig:CC-NM-Amp-Damp}
\end{figure}

\subsection{Thermal amplitude damping}
The non-Markovian amplitude damping can heuristically be extended to a finite-temperature case by coupling the qubit via a memory qubit to a Markovian thermal bath~\cite{heinekenQuantummemoryenhancedDissipativeEntanglement2021}. 
This leads to a second dissipative channel with operator \(\sigma_+\) and rate \(\gamma_+(t)\) in the master equation~(11).
\begin{align}
\label{eq:nMadthLindbladThermal}
\dot \rho=\mathcal{L}_t\left[\rho\right] = &\frac{\gamma_-(t)}{2} \left(\left[\sigma_- \rho, \sigma_+\right] + \left[\sigma_-, \rho  \sigma_+\right]\right) \notag\\
+&\frac{\gamma_+(t)}{2} \left(\left[\sigma_+ \rho, \sigma_-\right]  + \left[\sigma_+, \rho  \sigma_-\right]\right).
\end{align}
This equation models two competing non-Markovian amplitude damping processes with either the ground or the excited state as the fixed point. 
The analytical expressions for \(\gamma_\pm(t)\) is too lengthy and of no importance here. Their graphs are shown in Fig.~\ref{fig:gammapm} for parameters \(\gamma_0 = \alpha = 1\) and inverse temperature  \(\beta = 3.66\).
The rates flip sign at poles, a typical behavior for non-Markovian quantum dynamics.
\begin{figure}
    \centering
    \includegraphics[width=\columnwidth]{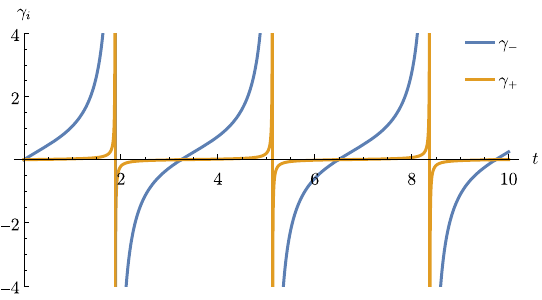}
    \caption{Time dependent rates \(\gamma_\pm\) of the time-local master equation \eqref{eq:nMadthLindbladThermal} for parameters \(\gamma_0 = \alpha =1\) and inverse temperature \(\beta = 3.66\). Both rates show poles and change sign over time. This is a typical behavior for non-Markovian quantum dynamics.} 
    \label{fig:gammapm}
\end{figure}

In Fig.~\ref{fig:CC-NM-Amp-Damp-Thermal} we show \(\coa\) and \(\cof\) as a function of time for the same choice of parameters. In this finite-temperature case, we see that \(\coa \neq \cof\).
In particular, for the chosen temperature and higher, the minima of \(\coa\) are above the highest maximum of \(\cof\). Thus, we have \(\coa\geq \cof\) for all times and the quantumness of the memory cannot be verified by our criterion in Thm.~1.
\begin{figure}
    \centering
    \includegraphics[width = \columnwidth]{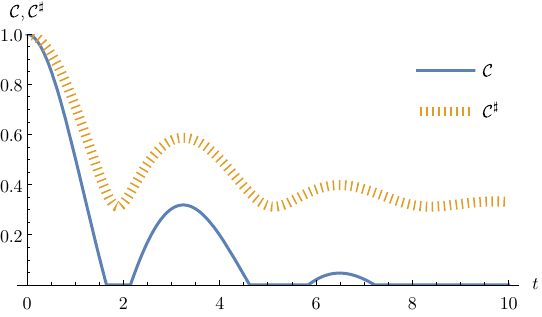}
    \caption{The plot shows the concurrence of assistance \(\coa\) and the concurrence of formation \(\cof\) of the Choi state of the non-Markovian thermal amplitude damping for an inverse temperature \(\beta = 3.66\).
    Both \(\coa\) and \(\cof\) still show revivals. However, the minima of \(\coa\) are above the maxima of \(\cof\). Thus, we have \(\coa(t_1) \geq \cof(t_2)\) for all choices of \(t_2>t_1\). Therefore, a quantum memory cannot be witnessed in this finite-temperature scenario. However, since a representation with classical memory is also missing, we cannot decide which kind of memory is required for the given dynamics.}
    \label{fig:CC-NM-Amp-Damp-Thermal}
\end{figure}

\section{Obtaining maps and master equations from a measurement prescription}

Here we will show how the genuinely non-Markovian master equation in Eq.~(13) emerges from a quantum jump trajectory representation with classical memory. 
The procedure described here by means of an example can straightforwardly be extended to more complex quantum jump scenarios.

The basic strategy is depicted in Fig.~\ref{fig:grundmodell}. 
There are two different measurements schemes \(\mathcal{M}\) and \(\tilde{\mathcal{M}}\) which represent different maps.
The first scheme \(\mathcal{M}\) is continuously applied until outcome \(1\) is observed. This outcome signifies a jump in the system. 
After this jump, the measurement scheme is switched and \(\tilde{\mathcal{M}}\) is applied for all times. 

\begin{figure}
    \centering
    \begin{tikzpicture}[scale=0.7]
	\draw (-0.5, 0) node{$\rho_S(0)$};
	\qwire{0}{0}{0.5}
	\measurement{0.5}{0}
	\qwire{1.5}{0}{0.5}
	\draw (2.5, 0) node{?};
	\memoryarrow{1}{-0.35}{1.5}
	\qwireup{3}{0}{4.5}{0.8}
	\qwireup{3}{0}{4.5}{-0.8}
	\qwire{4.5}{0.8}{0.5}
	\qwire{4.5}{-0.8}{0.5}
	\measurement{5}{0.8}
	\measurement{5}{-0.8}
	\memoryarrow{5.5}{0.45}{1.5}
	\draw (7, 0.8) node{?};
	\qwire{6}{-0.8}{1}	
	%\node at (7.5, 0.8) {$\dots$};
	\node at (7.5, -0.8) {$\dots$};
	\node [] (71) at (1, 0.7) {$\mathcal{M}$};
	\node [] (72) at (5.5, 1.5) {$\mathcal{M}$};
	\node [] (73) at (5.5, -1.5) {$\tilde{\mathcal{M}}$};
	\node [] (74) at (3.7, 0.8) {$M_2$};
	\node [] (75) at (3.7, -0.8) {$M_1$};
	\end{tikzpicture}
    \caption{Scheme of a measurement prescription with classical memory. Measurement \(\mathcal{M}\) is performed as long as it yields outcome \(2\) which is associated with the measurement operator \(M_2\). Once the measurement yields outcome \(1\) (with measurement operator \(M_1\), the measurement scheme is changed and a measurement \(\tilde{\mathcal{M}}\) is applied for all subsequent time steps.}
    \label{fig:grundmodell}
\end{figure}
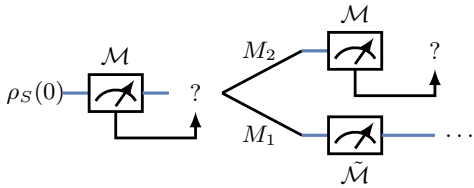

In this example, each measurement scheme has a Kraus representation of two measurement operators.
The first one \({\mathcal{M}}\) describes an amplitude damping process with $\ket{0}$ as its fixed point and is given by
\begin{align}
\label{eq:defK1K2}
%K_1 &= -\i \sin(\sqrt{\damp \Delta t}) \sigma_-,
M_1 &= -\i \sqrt{\damp \d t} \sigma_-,
&&
%K_2 = \cos(\sqrt{\damp \Delta t}) \sigma_+ \sigma_- + \sigma_- \sigma_+,
M_2=\id - \frac{\damp \d t}{2} \sigma_-\sigma_+,
\end{align}
where $\d t$ is an infinitesimal time step and $\damp$ a damping rate.
Measurement scheme $\tilde{M}$, which describes a damping towards \(\ket{1}\), takes the form
\begin{align}
\label{eq:defKAKB}
\tilde{M}_1 &= -\i \sqrt{\damp \d t} \sigma_+,
&&\tilde{M}_2 = \id - \frac{\damp \d t}{2}\sigma_+ \sigma_-.
\end{align}

Thus, a quantum jump representation of an amplitude damping channel is applied until a jump occurs. Once the jump happens (given by measurement operator \(M_1\)), the direction of the damping is reversed and an amplitude damping channel towards the excited state is applied for all times.

% We can now define a measurement prescription in the following way. First we start by applying measurement scheme $\mathcal{M}_1$ and note the outcome which corresponds to one of the two Kraus operators. Depending on this outcome there are two possibilities
% \begin{enumerate}
%     \item If the outcome of the measurement is the one corresponding to the application of $M_2$, we continue measuring with measurement scheme $\mathcal{M}_1$.
%     \item Whenever the outcome corresponds to $M_1$, we switch to the measurement scheme $\mathcal{M}_2$. From now on it is not necessary to keep track of the future outcomes, in any step measurement scheme $\mathcal{M}_2$ will be applied.
% \end{enumerate}
% \todo{image necessary? Which?}
This measurement prescription requires only one classical bit to be stored in order to describe the future dynamics of the system. This bit has to encode whether the operator $M_1$ has already been applied to the system in the history, i.e., whether a jump has already occurred.

Let us calculate the dynamics which is generated by this construction.
To obtain the map \(\cpt_t\) at time \(t\), we have to integrate over all possible trajectories weighted by their probabilities.
We can distinguish two different types of trajectories. There is a single one which does not contain any jump, i.e., where the measurement operator \(M_2\) is applied for all times. 
All other trajectories contain exactly one jump with operator \(M_1\) at a certain time \(t'\) and a subsequent evolution according to measurement scheme \(\tilde{\mathcal{M}}\). 

As long as outcome \(2\) occurs, the state evolves according to the following differential equation
\begin{align}
    \rho(t + \d t) = M_2 \rho(t) M_2^\dagger.
\end{align}
The dynamical map \(\M_2(t)\) which solves this equation for any input state is completely positive but not trace-preserving.
Furthermore, we denote the map for the single quantum jump by
\begin{align}
    \M_1[\rho] = M_1 \rho M_1^\dagger.
\end{align}

After the jump has occurred, the outcomes no longer matter, since we are interested in the reduced dynamics here. 
Thus, the dynamics after a jump is given by the equation
\begin{align}
    \rho(t+\d t) = \tilde{M}_1 \rho(t) \tilde{M}_1^\dagger + \tilde{M}_2 \rho(t) \tilde{M}_2^\dagger,
\end{align}
and we denote the resulting CPT dynamical map by \(\tilde{\M}(t)\).

The dynamics in \(\Sys\) for the given trajectory construction with classical memory is then given by
\begin{align}
    \cpt_t = \M_2(t) + \int_{0}^{t} \tilde{\M}(t-t') \circ \M_1 \circ \M_2(t') \d t'.
\end{align}
The first term on the right-hand side represents the trajectory without any jump.
The second one integrates over all trajectories that involve a jump. All these trajectories first evolve up to time \(t'\) according to the no-jump evolution \(\M_2(t')\). At time \(t'\), the jump happens, i.e., the map \(\M_1\) is applied. After the jump, the system evolves according to the dynamical map \(\tilde{\M}(t-t')\) up to time \(t\).
In this qubit case, the integration can conveniently be done in a Bloch vector representation~\cite{bethruskaiAnalysisCompletelypositiveTracepreserving2002}.

We can then bring the explicitly time-dependent generator of this dynamics \(\mathcal{G} = \dot{\cpt}_t\circ \cpt_t^{-1}\) into the form of a GKSL-type equation~\cite{zimanDescriptionQuantumDynamics2005,zimanOpenSystemDynamics2010a,hallCanonicalFormMaster2014} and obtain
\begin{align}
\label{eq:basicmodellindblad2}
\mathcal{L}_t\left[\rho\right] = \frac12 \sum_{k=1, 2} \gamma_k(t)  \left(\left[L_k, \rho L_k^\dagger \right]+\left[L_k\rho, L_k^\dagger\right] \right),
\end{align}
with 
\begin{align}
\gamma_1(t) &=\frac{\damp  (\damp  t-1)}{2 (\damp  t- \e^{\damp  t})},
&&\hphantom{L_2}\gamma_2(t)=\frac{\damp  \left(\e^{\damp  t}-1\right)}{8 \left(\e^{\damp  t}-\damp
	t\right)},\\
L_1 &= \sigma_-,
&&\hphantom{\gamma_2(t)}L_2 = {\sigma_z}.
\end{align}
The instantaneous rate \(\gamma_1(t)\) becomes negative for times \(t > 1/\kappa\). The dynamics described by this master equation is genuinely non-Markovian and, in particular, P-indivisible beyond this critical time.

\section{Relation to a physically measurable pure-state trajectory picture}
Markovian quantum dynamics can be unraveled into a set of pure-state trajectories that are physically measurable~\cite{wisemanQuantumMeasurementControl2009,brunSimpleModelQuantum2002}.
This means that there is a continuous measurement which keeps the system in a pure state at all times. Motivated by experimental setups, this continuous measurement is often constructed as an indirect monitoring of the system by measuring the environment.
However, the continuous measurement can equally be formulated on the system alone.

For non-Markovian dynamics the situation is less clear and there has been a debate on whether such physically measurable trajectories can exist~\cite{diosiNonMarkovianContinuousQuantum2008,wisemanPureStateQuantumTrajectories2008,breuerGenuineQuantumTrajectories2004,pellegriniNonMarkovianQuantumRepeated2009,kronkeNonMarkovianQuantumTrajectories2012,smirneRateOperatorUnraveling2020,megierContinuousQuantumMeasurement2020}. 
Many non-Markovian quantum trajectory constructions have been developed as powerful mathematical and numerical tools to approach the dynamics of open quantum systems in terms of pure states without a continuous measurement interpretation
~\cite{diosiNonMarkovianStochasticSchrodinger1997,breuerStochasticWavefunctionMethod1999,jackNonMarkovianQuantumTrajectory1999,piiloNonMarkovianQuantumJumps2008,hartmannExactOpenQuantum2017,gasbarriStochasticUnravelingsNonMarkovian2018,linkNonMarkovianQuantumDynamics2022a,beckerQuantumTrajectoriesTimelocal2023}. Thus, even though these frameworks decompose the reduced dynamics of the system into trajectories, those could not necessarily be realized in a time-continuous way by an actual measurement apparatus. 

The time-continuous limit of Eq.~(12) is the most general physically measurable quantum trajectory theory.   
In the Markovian case, the construction reduces to a pure-state unraveling of the dynamics, where the infinitesimal map connecting two times \(\cpt_{t+\mathrm{d}t,t}\) does not depend on previous outcomes~\cite{wisemanQuantumMeasurementControl2009}. Different decompositions of the same map \(\cpt_{t+\mathrm{d}t,t}\) lead to different unravelings of the dynamics, e.g., a jump unraveling or a diffusive unraveling. 
In so-called adaptive measurement schemes, this Kraus decomposition may depend on previous outcomes~\cite{wisemanQuantumMeasurementControl2009, karasikHowManyBits2011,daryanooshQuantumJumpsAre2014,beyerCollisionmodelApproachSteering2018}. 
Such adaptive measurements are naturally contained in  Eq.~(12).

The trajectory construction in Eq.~(12) is, however, far more general. It does not only allow the Kraus decomposition of the given reduced dynamics \(\cpt_{t+\mathrm{d}t,t}\) to depend on previous outcomes, but the map itself (denoted by \(\Phi_i\) in the two-step dynamics in Eq.~(3)) changes depending on the previously measured signal, as in feedback scenarios~\cite{wisemanQuantumMeasurementControl2009}.
Thus, the continuous case of Eq.~(12)) is the most general form of a physically measurable pure-state trajectory representation which includes the standard Markovian unravelings as well as adaptive measurement schemes and feedback.
Crucially, the average dynamics can be non-Markovian. 
Thus, non-Markovian dynamics that do not require truly quantum memory, have a physically measurable pure-state trajectory representation. On the other hand, if a dynamics is verified to rely on a truly quantum memory, such a trajectory picture is ruled out. 
 
% Combining these two steps we obtain a map which is close to the identity map \todo{(How to) Explain without Bloch representation?} which we will call $\Lambda{p_1, \temp}^\text{max}$. It becomes closer to the identity the smaller $p_1$ (small damping) and the closer $\temp$ is to zero (small shift).
% We can now start from this map and concatenate it with a map $\cpt{0}$ (which is not automatically a CPT map) with the aim of ending up with an overall map in the one-parameter family $\maptad{p}$.
% To check whether such a map can also be a proper quantum channel, we can use the Choi state. The Choi state describing the evolution from $\Lambda_{(p_1, \temp)}^\text{max}$ to the map describing the thermal amplitude damping $\maptad{p_2, \temp}$ we called $\cpt_0 = \maptad{p_2, \temp} \left({\Lambda_{(p_1, \temp)}^\text{max}}\right)^{-1}$ is given by
% \todo{reformulate with $\beta$}
% \begin{align}
% \choi(\cpt_0) = 
% \begin{pmatrix}
% \frac{1}{2} p_2 (\temp-1)+1 & 0 & 0 & \frac{2 \sqrt{1-p_2}}{p_1
% 	\sqrt{1-\temp^2}+2 \sqrt{1-p_1}} \\
% 0 & \frac{1}{2} p_2 (1+\temp) & 0 & 0 \\
% 0 & 0 & \frac{1}{2} p_2 (1-\temp) & 0 \\
% \frac{2 \sqrt{1-p_2}}{p_1 \sqrt{1-\temp^2}+2 \sqrt{1-p_1}} 
% & 0 & 0 & 1-\frac{1}{2} p_2 (\temp+1) \\
% \end{pmatrix}.
% \end{align}
% This map $\choi(\cpt_0)$ has to be a quantum channel which means that all eigenvalues have to be positive.
% \todo{Graphics?}

\end{document}